\documentclass[twoside]{article}
\usepackage{qic2}

\usepackage{bm}
\usepackage{amsfonts}
\usepackage{amscd}
\usepackage{amsmath}    
\usepackage{enumerate}

\newcommand{\bra}[1]{\mbox{$\left\langle #1 \right|$}}
\newcommand{\ket}[1]{\mbox{$\left| #1 \right\rangle$}}
\newcommand{\braket}[2]{\mbox{$\left\langle #1 | #2 \right\rangle$}}

\def\tr{{\rm Tr}}
\def\IR{{\mathbb R}}
\def\IC{{\mathbb C}}

\newcommand{\maxnorm}[1]{\left\lVert #1 \right\rVert_{\text{max}}}

\DeclareMathOperator{\myUgrp}{U}
\DeclareMathOperator{\myRe}{Re}

\def\dmathX#1#2{
$$\lineskiplimit=1000pt \advance\lineskip by #1\jot 
\mathsurround=0pt \tabskip=0pt plus 1000pt
\everycr{\noalign{\penalty\interdisplaylinepenalty}}
\halign to \displaywidth{
\hfil$\displaystyle{##}$\tabskip=0pt&%
\hfil $\displaystyle{{}##{}}$\hfil &%
\hfil $\displaystyle{{}##{}}$\hfil &%
$\displaystyle{##}$\hfil \tabskip=0pt plus 1000pt minus 1000pt&%
\refstepcounter{equation}\label{##}\llap{(\theequation)}\tabskip=0pt\cr
\noalign{\ifdim \prevdepth>-1000pt \vskip -#1\jot\fi}
#2\crcr}$$}

\newenvironment{proof}{\noindent{\bf Proof:}}{\squareQIC}

\begin{document}
\setlength{\textheight}{8.0truein}    

\runninghead{Solution to time-energy costs of quantum channels}
            {C.-H. F. Fung, H. F. Chau, C.-K. Li, and N.-S. Sze}

\normalsize\textlineskip
\thispagestyle{empty}
\setcounter{page}{1}


\vspace*{0.88truein}

\alphfootnote

\fpage{1}

\centerline{\bf
SOLUTION TO TIME-ENERGY COSTS OF QUANTUM CHANNELS}
\vspace*{0.37truein}
\centerline{\footnotesize
CHI-HANG FRED FUNG\footnote{Present address: Huawei Noah's Ark Lab, Hong Kong Science Park, Shatin, Hong Kong.  E-mail: chffung.app@gmail.com}}
\vspace*{0.015truein}
\centerline{\footnotesize\it 
Department of Physics and Center of Theoretical and Computational Physics, }
\baselineskip=10pt
\centerline{\footnotesize\it University of Hong Kong, Pokfulam Road, Hong Kong}
\vspace*{10pt}
\centerline{\footnotesize 
H.~F. CHAU}
\vspace*{0.015truein}
\centerline{\footnotesize\it 
Department of Physics and Center of Theoretical and Computational Physics, }
\baselineskip=10pt
\centerline{\footnotesize\it University of Hong Kong, Pokfulam Road, Hong Kong}
\vspace*{10pt}
\centerline{\footnotesize 
CHI-KWONG LI}
\vspace*{0.015truein}
\centerline{\footnotesize\it Department of Mathematics, College of William \& Mary, Williamsburg, Virginia 23187-8795, USA}
\vspace*{10pt}
\centerline{\footnotesize 
NUNG-SING SZE}
\vspace*{0.015truein}
\centerline{\footnotesize\it Department of Applied Mathematics, The Hong Kong Polytechnic University, Hung Hom, Hong Kong}
\vspace*{0.225truein}

\vspace*{0.21truein}

\abstracts{
We derive a formula for the time-energy costs of general quantum channels proposed in
[Phys. Rev. A {\bf 88}, 012307 (2013)].
This formula allows us to numerically find the time-energy cost of any quantum channel using positive semidefinite programming.
We also derive 
a lower bound to the time-energy cost for any channels and
the exact 
the time-energy cost for a class of channels which 
includes the qudit depolarizing channels and projector channels as special cases.
}{}{}

\vspace*{10pt}

\vspace*{3pt}
\communicate{to be filled by the Editorial}

\vspace*{1pt}\textlineskip 

%
%

%
%
%

\section{Introduction}

A time-energy cost of a unitary matrix 
$U \in \myUgrp(r)$
is defined as~\cite{Chau2011}
\begin{align}
\label{eqn-definition-maxnorm-for-U}
\maxnorm{U}&=\max_{1 \le j \le r}  |\theta_j|
\end{align}
where $U$ has eigenvalues 
$\exp(i \theta_j)$ for $j=1,\dots,r$.
Here, we
denote by $\myUgrp(r)$ the group of $r \times r$ unitary matrices,
and
we take the convention that $\theta_j \in (-\pi,\pi]$.
This definition of time-energy cost was motivated~\cite{Chau2011,Fung:2013:Time-energy} from time-energy uncertainty relations~\cite{Lloyd2000,Mandelstam1945}.
Essentially, this time-energy cost captures the idea that time and energy are a trade-off against each other and may be used as an indicator for the resource used by a quantum system.
In particular,
a closed quantum system with a time-independent Hamiltonian $H$ evolves from the initial state  $\ket{\psi_\text{i}}$ to the final state $\ket{\psi_\text{f}}$ according to the Schr\"{o}dinger equation:
$\ket{\psi_\text{f}} = U \ket{\psi_\text{i}}$
where
$U=\exp(-i Ht/\hbar)$ and $t$ is the evolution time.
The eigenvalues of the Hamiltonian $H$ are the energies and thus the eigenvalues of $\log U$ correspond to the time-energy products, the absolute maximum of which is the time-energy cost $\maxnorm{U}$ defined above.
Note that to implement the same information processing task characterized by $U$, one may use a high energy $H$ run for a short time or a low energy $H$ run for a long time.
The time-energy products in both cases are the same.

The definition for $\maxnorm{U}$  in Eq.~\eqref{eqn-definition-maxnorm-for-U} is for unitary quantum channels.
The time-energy cost has been extended to cover general quantum channels~\cite{Fung:2013:Time-energy}.
A quantum channel mapping $n$-dimensional density matrices to $n$-dimensional density matrices can be 
written as
\begin{align}
{\mathcal K}(\rho)=\sum_{j=1}^d K_j \rho K_j^\dag,
\end{align}
where $K_j \in \IC^{n\times n}$ are the Kraus operators and $\sum_{j=1}^d K_j^\dag K_j=I_n$.
In this paper, we only consider finite dimensional systems.
The time-energy cost for quantum channel $\mathcal K$ is defined as
the time-energy cost of the most efficient unitary extension that implements $\mathcal K$~\cite{Fung:2013:Time-energy}:
\dmathX2{
\maxnorm{\mathcal{K}} &\equiv& \min_U & \maxnorm{U}  
&eqn-energy-measure-general-channel\cr
&&
\text{s.t.} &
\mathcal{K}(\rho) = \tr_B [ U_{BA} (\ket{0}_B\bra{0} \otimes  \rho_A) U_{BA}^\dag ]
\: \forall \rho,
\cr
}
where the channel $\mathcal{K}$ acts on quantum state $\rho$ in system $A$ and the unitary extension $U_{BA}$ includes system $B$ prepared in a standard state.

The time-energy cost has an interesting informational meaning.
The cosine of this cost for a general quantum channel is exactly the worst-case entanglement fidelity of the channel~\cite{Fung:2014:fidelity}, 
establishing a connection between the physical aspect (the time-energy cost) and the information aspect (the fidelity) of quantum channels.
Fidelity is a popular quantity often used to characterize the performance of information processing tasks including
quantum key distribution (as
a security measure~\cite{Konig:2007:AccessibleInformation,Ben-Or:2005:Composable}) and
state discrimination (as the inconclusive probability~\cite{Ivanovic:1987:USD,Dieks:1988:USD,Peres:1988:USD}).
Thus the study of the time-energy cost is important from a quantum information theoretical perspective.
To be specific, the result of Ref.~\cite{Fung:2014:fidelity} shows that for any quantum channel $\mathcal K$, the worst-case entanglement fidelity $F_\text{min}({\mathcal K})$ of the channel is related to the time-energy cost by\footnote{Note that Ref.~\cite{Fung:2014:fidelity} originally shows that 
$F_\text{min}({\mathcal K})=
\max(
\cos \maxnorm{\mathcal K},0
)$.  However, we should always consider taking the freedom of including an all-zero Kraus operator in the channel representation.
In this case, $\cos \maxnorm{\mathcal K}$ is never negative. See Theorem~\ref{thm-TE-general-solution} and its proof.}
\begin{align}
\label{eqn-main-theorem-fidelity}
F_\text{min}({\mathcal K})
=
\cos \maxnorm{\mathcal K} .
\end{align}
Here, the worst-case entanglement fidelity $F_\text{min}({\mathcal K})$ is defined as
\begin{align}
\label{eqn-def-fidelity-channel}
F_\text{min}({\mathcal K})
\equiv
\min_{\ket{\Psi}}
F\big(\ket{\Psi}_{AC}\bra{\Psi},({\mathcal K}_A \otimes I_C)(\ket{\Psi}_{AC}\bra{\Psi})\big),
\end{align}
where
the channel acts on system $A$ and the fidelity is taken between the channel input state (allowed to be entangled in systems $A$ and $C$) and the corresponding output state.
Here, $F(\rho,\rho')\equiv\tr \sqrt{\rho^{1/2} \rho' \rho^{1/2}}$
is the fidelity between two mixed quantum states $\rho$ and $\rho'$~\cite{Jozsa1994,Uhlmann1976}.

This paper derives a formula for the time-energy cost $\maxnorm{\mathcal K}$ defined in Eq.~\eqref{eqn-energy-measure-general-channel} and provides a numerical solution method via semidefinite programming.
This in turn allows us to compute the the worst-case entanglement fidelity using Eq.~\eqref{eqn-main-theorem-fidelity}.
The difficulty in solving for $\maxnorm{\mathcal K}$ stems from the freedom in the unitary extension.
All the freedom we have for choosing different $U$ without changing the channel consists of 
the following operations:
\begin{enumerate}
\item Change the last $(d+1)n-n$ columns of $U$.
\item Apply $V\otimes I_n$ to $U$ on the left, where $V \in \myUgrp(d+1)$.
\end{enumerate}
It turns out that one can apply an abstract mathematical result in unitary
dilation theory~\cite{Choi:2001} to solve the problem.
One can then determine the
optimal solution using semidefinite programming. Thus, we have a
theoretical optimal solution that can be determined by numerical method.
This is one of the best scenarios in solving an optimization problem
if there is a closed form for the optimal solution of the given problem.

The organization of this paper is as follows.
We solve problem~\eqref{eqn-energy-measure-general-channel} for $\maxnorm{\mathcal{K}}$ in Sec.~\ref{sec-main-result}, and we derive a lower bound to the time-energy cost for any channels and compute the exact time-energy costs for special channels in Sec.~\ref{sec-special-channels}.
We formulate in Sec.~\ref{sec-sdp} the problem of finding the time-energy cost as a semidefinite program (SDP) which can be solved numerically and efficiently.
We give some mathematical remarks in Sec.~\ref{sec-remarks} and conclude in Sec.~\ref{sec-conclusions}

\section{Main result}
\label{sec-main-result}

\begin{theorem}
\label{thm-TE-general-solution}
{\rm
\begin{align}
\label{eqn-TE-general-solution}
\maxnorm{\mathcal{K}} =
\cos^{-1}
\left[
\max_{\mathbf v}
\frac{1}{2} \lambda_\text{min} \left( K_{\mathbf v} + K_{\mathbf v}^\dag \right)
\right]
\end{align}
where ${\mathbf v} \in \IC^d$ has 
$\ell_2$-norm $\lVert \mathbf v \rVert \le 1$,
$K_{\mathbf v}=\sum_{j=1}^d v_j K_j$, $\lambda_\text{min}(\cdot)$ denotes the minimum eigenvalue of its argument,
and
we take the convention that $\cos^{-1}$ returns an angle in the range $[0,\pi]$.
}
\end{theorem}
\begin{proof}
The most general form of $U$
in Eq.~\eqref{eqn-energy-measure-general-channel} is
\begin{align}
\label{eqn-U-general-form}
U=(V \otimes I_n) 
\underbrace{
\begin{bmatrix}
K_1 & * & * & \cdots & *
\\
K_2 & * & * & \cdots & *
\\
\vdots & & & & \vdots
\\
K_{d} & * & * & \cdots & *
\\
K_{d+1} & * & * & \cdots & *
\end{bmatrix}
}_{\displaystyle U'}
\end{align}
where $V \in \myUgrp(d+1)$ and only the first $n$ columns of $U'$ are fixed.
Here, we append an all-zero Kraus operator $K_{d+1}=0$ in order to make $U$ the most general unitary implementing the channel $\mathcal K$.
Certainly, both $\{K_1,\dots,K_d\}$ and $\{K_1,\dots,K_{d+1}\}$ are valid representations of $\mathcal K$. 
As we shall see, there is no need to add more than one extra all-zero operator.

We first consider the freedom in $U'$.
Let $d'=d+1$.
We want to choose the last $d'n-n$ columns of $U'$ so that its norm is the smallest.
This is described as an optimization problem as follows:
\dmathX2{
\varphi
&\equiv& \displaystyle\min_{U'} & 
\maxnorm{ U' }\cr
&&\text{s.t.}&
U'_{i1}=K_i \:\: \text{for all }i=1,\ldots,d',\cr
&&&\text{with }U' \in \myUgrp(d'n)&eqn-problem-original-min-U\cr
}
where $U'_{ij}$ denotes the $(i,j)$ block of size $n \times n$.

By the result in Ref.~\cite{Choi:2001}, we know 
that there is a unitary matrix $\tilde U = (\tilde U_{rs})_{1\le r,s \le 2} \in \myUgrp(2n)$ 
with eigenvalues
$e^{\pm i\theta_j}$ for $j = 1, \dots, n$, 
such that 
$\tilde U_{11}=K_1$ and 
$\tilde U_{21}=\sqrt{I_n - K_1^\dag K_1}$
where $\pi \ge \theta_1 \ge \cdots \ge \theta_n\ge 0$ and
$\cos(\theta_1) = \lambda_\text{min}(K_1+K_1^\dag)/2$.
Note that there exists $W \in \myUgrp(d'n-n)$ such that
$(I_n \oplus W) (\tilde U \oplus I_{d'n-2n})(I_n \oplus W)^\dag$ satisfies the constraints in Eq.~\eqref{eqn-problem-original-min-U} and 
thus 
\begin{align}
\varphi \le
\maxnorm{\tilde{U}} 
=
\cos^{-1}
\left[
\frac{1}{2} \lambda_\text{min} \left( K_{1} + K_{1}^\dag \right)
\right].
\label{eqn-U-upper-bound1}
\end{align}

Next, 
we lower bound $\varphi$.
Consider $U'$ satisfying the constraints in Eq.~\eqref{eqn-problem-original-min-U}.
By the interlacing inequalities (see, e.g., Ref.~\cite{Fan1957}), 
because $(K_1+K_1^\dag)/2$ is the principal submatrix of $(U'+U'^\dag)/2$,
the eigenvalues $a_1\ge\dots\ge a_{d'n}$ of $(U'+U'^\dag)/2$ and the eigenvalues $b_1\ge\dots\ge b_n$ of $(K_1+K_1^\dag)/2$ satisfy
$$
a_{d'n} \le b_n \le a_n,
$$
and so
$$
\cos^{-1}(a_{d'n}) \ge \cos^{-1}(b_n).
$$
If $U'$ has eigenvalues $\exp(i \theta_j)$, where $j=1,\dots,d'n$ and $\theta_j \in (-\pi,\pi]$,
then $a_{d'n}=\cos(\max_j |\theta_j|)$, giving
$$
\max_j |\theta_j|
\ge
\cos^{-1}
\left[
\frac{1}{2} \lambda_\text{min} \left( K_{1} + K_{1}^\dag \right)
\right].
$$
Thus, \eqref{eqn-problem-original-min-U} is bounded by
\begin{align}
\varphi
\ge 
\cos^{-1}
\left[
\frac{1}{2} \lambda_\text{min} \left( K_{1} + K_{1}^\dag \right)
\right].
\end{align}
Combining with Eq.~\eqref{eqn-U-upper-bound1} gives
\begin{align}
\varphi
=
\cos^{-1}
\left[
\frac{1}{2} \lambda_\text{min} \left( K_{1} + K_{1}^\dag \right)
\right].
\end{align}

Finally, we optimize $V$ in Eq.~\eqref{eqn-U-general-form} to obtain $\maxnorm{\mathcal K}$.
Note that $\varphi$ which corresponds to the optimal solution of $U'$ after adjusting the last $d'n-n$ columns depends only on the principal submatrix of $U'$.
Thus, 
\begin{align}
\maxnorm{\mathcal K}
&=
\cos^{-1}
\left[
\max_{\mathbf{v}: \: \lVert \mathbf{v} \rVert = 1}
\frac{1}{2} \lambda_\text{min} \left( K_{\mathbf v} + K_{\mathbf v}^\dag \right)
\right]
\end{align}
where $\mathbf v \in \IC^{d+1}$ is the first row of $V$.
Here, $K_{\mathbf v}=\sum_{j=1}^{d+1} v_j K_j$ represents the principal submatrix of $U$, where $\mathbf v=[v_1,\dots,v_{d+1}]$.
Taking into account $K_{d+1}=0$ gives the claim of the theorem.
\end{proof}

We remark that $\cos \maxnorm{\mathcal K} \ge 0$.

\section
{Time-energy costs for special channels}
\label{sec-special-channels}

In this section, we use Theorem~\ref{thm-TE-general-solution} to compute the time-energy costs for a class of channels which includes the qudit depolarizing channels and projector channels as special cases.

\begin{lemma}
\label{lemma-trace-zero-Kraus}
{\rm
Any channel $\mathcal K$ 
can be described by
an equivalent form with the
Kraus operators $\{ K_j \in \IC^{n \times n}: j=1,\ldots,d \}$ satisfying
\begin{align}
\tr({K}_j) &= 0, \:\: j=2,\ldots,d.
\nonumber
\end{align}
}
\end{lemma}
\begin{proof}
Two sets of Kraus operators $\{K_1,\dots,K_d\}$ and $\{\tilde{K}_1,\dots,\tilde{K}_d\}$ describe the same quantum channel if and only if
\begin{equation}
{K}_i=\sum_{j=1}^d w_{ij} \tilde{K}_j, \text{ for } i=1,\dots,d 
\label{eqn-equivalent-Kraus}
\end{equation}
and for some unitary matrix $W\equiv[w_{ij}]$ of dimension $d$ (see, e.g., Theorem~8.2 of Ref.~\cite{Nielsen2000}).
By taking the trace of Eq.~\eqref{eqn-equivalent-Kraus}, we see that there must exist $W$ that can bring $d-1$  terms to zero.
In particular, we have
\begin{equation}
\label{eqn-equiv-K1}
K_1=\left( \sum_{j=1}^d | \tr (\tilde{K}_j) | ^2 \right)^{-\frac{1}{2}} \sum_{j=1}^d \tr^\dag (\tilde{K}_j)  \tilde{K}_j .
\end{equation}
\end{proof}

(If $d=1$, 
we can pad the channel with $K_2=0$ to make Lemma~\ref{lemma-trace-zero-Kraus} automatically hold.)

\begin{lemma}
\label{lemma-trace-zero-Kraus-lower-bound}
{\rm
For any channel $\mathcal K$ that can be described by
Kraus operators $\{ K_j \in \IC^{n \times n}: j=1,\ldots,d \}$ of the form
\begin{align}
\tr({K}_j) &= 0, \:\: j=2,\ldots,d,
\nonumber
\end{align}
we have
\begin{equation}
\label{eqn-trace-zero-Kraus-lower-bound}
\cos^{-1} 
\left[
\frac{1}{n} \left\lvert \tr \left( K_1 \right) \right\rvert
\right]
\le
\maxnorm{\mathcal{K}} .
\end{equation}
}
\end{lemma}
\begin{proof}
We consider the middle term of Eq.~\eqref{eqn-TE-general-solution}:
\begin{align*}
\frac{1}{2} \lambda_\text{min} \left( K_{\mathbf v} + K_{\mathbf v}^\dag \right)
&
\le \frac{1}{2 n} \sum_{i=1}^n \lambda_i \left( K_{\mathbf v} + K_{\mathbf v}^\dag \right)
\\
&=
\frac{1}{2 n} \tr \left( K_{\mathbf v} + K_{\mathbf v}^\dag \right)
\\
&=
\frac{1}{n} \myRe \left[ \tr \left( K_{\mathbf v} \right) \right]
\\
&=
\frac{1}{n} \myRe \left[ v_1 \tr \left( K_1 \right) \right]
\end{align*}
where the first line is because the minimum is no greater than the average and $\lambda_i$ denotes the $i$th eigenvalue.
Maximizing over $\mathbf v$ gives the claim.
\end{proof}
\begin{theorem}[Time-energy lower bound]
\label{thm-lower-bound}
{\rm
For any channel $\mathcal K$ described by
Kraus operators $\{ K_j \in \IC^{n \times n}: j=1,\ldots,d \}$,
we have
\begin{equation}
\cos^{-1} 
\left[
\frac{1}{n} 
\sqrt{
\sum_{j=1}^d
\left\lvert \tr \left( K_j \right) \right\rvert^2
}
\right]
\le
\maxnorm{\mathcal{K}} .
\end{equation}
}
\end{theorem}
\begin{proof}
This follows from Lemma~\ref{lemma-trace-zero-Kraus} 
and Lemma~\ref{lemma-trace-zero-Kraus-lower-bound}.
\end{proof}

\begin{theorem}[Time-energy for special channels]
\label{thm-K1-is-I}
{\rm
For any channel $\mathcal K$ that can be described by
Kraus operators $\{ K_j \in \IC^{n \times n}: j=1,\ldots,d \}$ of the form
\begin{equation}
\label{eqn-channel-class1}
\begin{aligned}
{K}_1 &= \alpha I 
\text{ where $\alpha \in \IC$}
\\
\tr({K}_j) &= 0, \:\: j=2,\ldots,d,
\end{aligned}
\end{equation}
its time-energy cost is
\begin{equation}
\maxnorm{\mathcal{K}} =
\cos^{-1} 
\lvert \alpha \rvert .
\end{equation}
}
\end{theorem}
\begin{proof}
From Eq.~\eqref{eqn-trace-zero-Kraus-lower-bound}, we have
$\cos^{-1} |\alpha| \le \maxnorm{\mathcal K}$.

On the other hand, by choosing a particular $\mathbf v$,
\begin{align*}
&
\max_{\mathbf v}
\frac{1}{2} \lambda_\text{min} \left( K_{\mathbf v} + K_{\mathbf v}^\dag \right)
\\
\ge &
\max_{\theta_1}
\frac{1}{2} \lambda_\text{min} \left( e^{i \theta_1} K_1 + e^{-i \theta_1} K_1^\dag \right)
\\
= &
|\alpha| .
\end{align*}
Therefore, 
$\maxnorm{\mathcal K} \le \cos^{-1} |\alpha|$ and the claim is proved.
\end{proof}

Note that this theorem is slightly more general than Eq.~(52) of Ref.~\cite{Fung:2013:Time-energy} in which $\alpha$ is real and positive.
As noted in Ref.~\cite{Fung:2013:Time-energy}, channels satisfying Eq.~\eqref{eqn-channel-class1} include the qudit depolarizing channels.
In the following, we show that projector channels also satisfy Eq.~\eqref{eqn-channel-class1}.

In general, given a channel,
we can find an equivalent form according to Lemma~\ref{lemma-trace-zero-Kraus} and compute the new $K_1$ using Eq.~\eqref{eqn-equiv-K1}.
If this new $K_1$ satisfies Eq.~\eqref{eqn-channel-class1}, then the time-energy cost of the channel is immediately given by Theorem~\ref{thm-K1-is-I}.
Otherwise, we can lower bound it using Theorem~\ref{thm-lower-bound}.

\begin{theorem}[Projector channels]
{\rm
For any channel $\mathcal K$ that can be described by
Kraus operators $\{ K_j \in \IC^{n \times n}: j=1,\ldots,d \}$ of the form
$K_j=s_j P_j$ with $P_j=P_j^2=P_j^\dag$ being a projector of rank $r$ and $s_j \in \IC$,
we have
\begin{equation}
\maxnorm{\mathcal{K}} =
\cos^{-1} 
\left( \sqrt{\frac{r}{n}} \right) .
\end{equation}
}
\end{theorem}
\begin{proof}
Note that $\tr(K_j)=s_j r$ for all $j$.
Using Lemma~\ref{lemma-trace-zero-Kraus} and Eq.~\eqref{eqn-equiv-K1},
an equivalent description of $\mathcal K$ satisfies
\begin{align*}
K_1'&=\frac{1}{\sqrt{\sum_{i=1}^d |s_i|^2}} I,
\\
\tr({K}_j') &= 0, \:\: j=2,\ldots,d.
\end{align*}
Next, note that the trace-preserving constraint of quantum channels implies that $I_n=\sum_{j=1}^d K_j^\dag K_j = \sum_{j=1}^d |s_j|^2 P_j$ and taking the trace of it gives $n/r=\sum_{j=1}^d |s_j|^2$.
Then by Theorem~\ref{thm-K1-is-I}, the claim is proved.
\end{proof}

\section{Efficient numerical solution using semidefinite programming}
\label{sec-sdp}

Our main result \eqref{eqn-TE-general-solution} in Theorem~\ref{thm-TE-general-solution} can be formulated as an SDP.
We can write $K_j=A_j +i B_j$, where $A_j,B_j \in \IC^{n\times n}$ are Hermitian, and also write $v_j=a_j - i b_j$ with $a_j, b_j \in \IR$ for $j=1,\dots,d$.
Then the problem is equivalent to 
\begin{equation}
\begin{aligned}
\label{eqn-sdp-prob1}
\max \phantom{xxxx}
& \lambda_\text{min} \left(
\sum_{i=1}^d (a_j A_j + b_j B_j)
\right)
\\
\text{s.t.} \phantom{xxxx}& 
\sum_{j=1}^d (a_j^2 + b_j^2) \le 1
\end{aligned}
\end{equation}
where the maximization is over 
$a_1, b_1, \dots, a_d, b_d \in \IR$.
We show that this problem can be cast as a complex SDP which has the following form:
\begin{equation}
\label{eqn-standard-sdp}
\begin{aligned}
\min \phantom{xxxx}
& g^T x
\\
\text{s.t.} \phantom{xxxx}& x_1 G_1+\dots+x_mG_m +H \succeq 0
\end{aligned}
\end{equation}
where the minimization is over $x \in \IR^m$. Here, $g \in \IR^m$, and $G_1,\dots,G_m, H$ are complex 
Hermitian matrices.
Note that a complex SDP can always be cast as a real SDP in which $G_1,\dots,G_m, H$ are real symmetric matrices.

Note that we can rewrite the objective function as follows:
\begin{equation}
\begin{aligned}
\min \phantom{xxxx}
& -\lambda
\\
\text{s.t.} \phantom{xxxx}& \sum_{j=1}^d (a_j^2 + b_j^2)\le 1
\\
&
\sum_{i=1}^d (a_j A_j + b_j B_j) \succeq \lambda I
\end{aligned}
\end{equation}
where the maximization is over 
$a_1, b_1, \dots, a_d, b_d, \lambda \in \IR$.
Next, we convert this inequality constraint to a positive semidefinite constraint.
Let $c=\sqrt{\sum_{j=1}^d (a_j^2 + b_j^2)}$.
Consider the matrix
\begin{align*}
C=
\begin{bmatrix}
1&c\\
c&1
\end{bmatrix}
\end{align*}
which has eigenvalues $1\pm c$.
Thus, the constraint $c \le 1$ is equivalent to the constraint $C \succeq 0$.
Note that $C \oplus I_{2d-1}$ is unitarily similar to 
$$
a_1 F_1 + \dots a_d F_d+ b_1 F_{d+1} + \dots + b_d F_{2d} + I_{2d+1}
$$
where $F_j=E_{j,2d+1}+E_{2d+1,j}$ and $E_{i,j}$ is an $(2d+1)\times(2d+1)$ matrix with one at the $(i,j)$ position.
Then, the problem becomes
\begin{equation}
\begin{aligned}
\min \phantom{x}
& -\lambda
\\
\text{s.t.} \phantom{x}& 
a_1 F_1 + \dots a_d F_d+ b_1 F_{d+1} + \dots + b_d F_{2d} + I_{2d+1} \succeq 0
\\
&
\sum_{i=1}^d (a_j A_j + b_j B_j) - \lambda I \succeq 0
\end{aligned}
\end{equation}
where the maximization is over 
$a_1, b_1, \dots, a_d, b_d, \lambda  \in \IR$.
This is in the SDP form \eqref{eqn-standard-sdp}.
Thus, one can apply standard positive semidefinite programming to determine the time-energy cost of a general quantum channel
given in Eq.~\eqref{eqn-TE-general-solution}.

\section{Mathematical remarks}
\label{sec-remarks}
\begin{itemize}

\item
We may replace $K_1$
by $e^{i \theta_1} K_1$ without affecting the quantum channel.
Thus, we can select $\theta_1 \in [0,2\pi)$ to maximize the smallest eigenvalue of $e^{i \theta_1} K_1+e^{-i \theta_1} K_1^\dag$. 
To this end, we can use the
numerical range of $K_1$
defined as
$$
W(K_1)=\{ \bra{x} K_1 \ket{x} : \ket{x} \in \IC^n,  \braket{x}{x}=1 \} .
$$
This is a compact convex set in $\IC$, and can be obtained as the intersection of the half spaces
\begin{align*}
Q_{\theta_1}=\big\{ & \mu \in \IC: e^{i \theta_1}\mu + e^{-i \theta_1}\bar{\mu} \ge 
\\
&
\lambda_\text{min}(e^{i \theta_1}K_1+e^{-i \theta_1}K_1^\dag)\big\}, 
\:\:\:\: \theta_1 \in [0,2\pi) .
\end{align*}
So, maximizing the smallest eigenvalue of $e^{i \theta_1}K_1+e^{-i \theta_1}K_1^\dag$
corresponds to finding the half space
$Q_{\theta_1}$ whose intersection with the unit disk has the smallest area.

\item
A heuristic approach to upper bound Eq.~\eqref{eqn-TE-general-solution} is as follows.
We separately consider $v_j K_j, j=1,\dots,d$ and let $v_j=c_j \exp(i \theta_j)$ where $c_j \in \IR_+$.
Choose $\theta_j \in [0,2\pi)$ to
maximize the smallest eigenvalue $\sigma_j$ of $e^{i \theta_j} K_j+e^{-i \theta_j} K_j^\dag$.
This is equivalent to rotating the numerical range $W(K_j)$ so that the left support line is as close to the right side as possible.
Then choose a nonnegative unit vector $(c_1,\dots,c_d)$ to maximize $\sum_{j=1}^d c_j \sigma_j$.
If $K_{\mathbf v}=\sum_{j=1}^d c_j \exp(i \theta_j) K_j$, then $\lambda_\text{min} \left( K_{\mathbf v} + K_{\mathbf v}^\dag \right) \ge \sum_{j=1}^d c_j \sigma_j$.
Thus, $\maxnorm{\mathcal K} \le \cos^{-1} (\sum_{j=1}^d c_j \sigma_j/2)$.

\end{itemize}

\section{Conclusions}
\label{sec-conclusions}

The physical meaning of the time-energy cost is its relation with the channel fidelity~\cite{Fung:2014:fidelity}.
In this paper, we show that the time-energy cost of any general quantum channel is given by Eq.~\eqref{eqn-TE-general-solution}.
It has closed formulas for special channels.
For general channels, the problem of finding the time-energy cost can be formulated as an SDP which can be solved efficiently on computers.

\section*{Acknowledgments}%
Li would like to thank Mikio Nakahara for some inspiring discussion.
Chau and Fung were partially supported by the Hong Kong RGC grant
No. 700712P. 
Sze was partially supported by the Hong Kong RGC grant PolyU 502512.
Li
was supported
by a USA NSF grant and a Hong Kong RGC grant. 
He was visiting the Hong Kong Polytechnic University, the University of Hong Kong, the Shanghai University, and the Institute for Quantum Computing at the University of Waterloo in the Spring of 2014. He would like to thank the support and hospitality of colleagues of these institutions.

\vspace{50pt}

\nonumsection{References}
\bibliographystyle{IEEEtran2}
\bibliography{paperdb}

\end{document}